\documentclass{article}

\usepackage{graphicx} 
\usepackage{subfigure}

\usepackage{natbib}
\usepackage{algorithm}
\usepackage{algorithmic}

\usepackage{amsmath,amssymb,amsthm}

\usepackage[nohyperref,preprint]{icml2011}

\usepackage{comment}

\icmltitlerunning{Computational Rationalization: The Inverse Equilibrium Problem}

\newtheorem{definition}{Definition}
\newtheorem{lemma}{Lemma}
\newtheorem{property}{Property}

\newtheorem{corollary}{Corollary}
\newtheorem{theorem}{Theorem}

\newcommand{\defword}[1]{\textbf{\boldmath{#1}}}
\newcommand{\simplex}[1]{\Delta_{#1}}
\newcommand{\A}{\mathcal A}

\newcommand{\func}[3]{#1 : #2 \mapsto #3}
\newcommand{\R}{\mathbb R}
\newcommand{\ami}{a_{-i}}
\newcommand{\Ami}{\mathcal{A}_{-i}}
\newcommand{\norm}[1]{\left|\left|#1\right|\right|}
\newcommand{\argmax}{\operatornamewithlimits{argmax}}
\newcommand{\regret}[4]{\mathrm{regret}_{#1}(#2|#3,#4)}
\newcommand{\switch}[4]{\mathrm{switch}_{#1}^{#2\rightarrow #3}(#4)}

\newcommand{\inti}{\Phi_i^{\mathrm{int}}}
\newcommand{\modint}{\Phi^{\mathrm{int}}}
\newcommand{\swapi}{\Phi_i^{\mathrm{swap}}}
\newcommand{\swap}{\Phi^{\mathrm{swap}}}
\newcommand{\Regret}[3]{R^{#1}(#2,#3)}
\newcommand{\Rmat}[2]{R_{#1}^{#2}}
\newcommand{\baRmat}[2]{\bar{R}_{#1}^{#2}}
\newcommand{\st}{\mbox{~s.t.~}}
\newcommand{\trans}[1]{{#1}^\mathrm{T}}
\newcommand{\lambdaT}{\trans\lambda}
\newcommand{\truesigma}{\sigma}
\newcommand{\demonsigma}{\tilde\sigma}
\newcommand{\predsigma}{\hat\sigma}
\newcommand{\predw}{\hat{w}}
\newcommand{\rvec}[3]{r^{#2}_{#1,#3}}
\newcommand{\barvec}[3]{\bar{r}^{#2}_{#1,#3}}

\begin{document}

\twocolumn[
\icmltitle{Computational Rationalization: The Inverse Equilibrium Problem}

\icmlauthor{Kevin Waugh}{waugh@cs.cmu.edu}
\icmlauthor{Brian D. Ziebart}{bziebart@cs.cmu.edu}
\icmlauthor{J. Andrew Bagnell}{dbagnell@ri.cmu.edu}
\icmladdress{Carnegie Mellon University, 5000 Forbes Ave, Pittsburgh, PA, USA 15213}

\icmlkeywords{game theory, multi-agent, imitation learning, control}

\vskip 0.3in
]

\newcommand{\todo}[1]{}

\begin{abstract}
Modeling the purposeful behavior of imperfect agents from a small number of observations is a challenging task.  When restricted to the single-agent decision-theoretic setting, inverse optimal control techniques assume that observed behavior is an approximately optimal solution to an unknown decision problem.  These techniques learn a utility function that explains the example behavior and can then be used to accurately predict or imitate future behavior in similar observed or unobserved situations.

In this work, we consider similar tasks in competitive and cooperative multi-agent domains.  Here, unlike single-agent settings, a player cannot myopically maximize its reward --- it must speculate on how the other agents may act to influence the game's outcome.  Employing the game-theoretic notion of regret and the principle of maximum entropy, we introduce a technique for predicting and generalizing behavior, as well as recovering a reward function in these domains.
\end{abstract}

\section{Introduction}



Predicting the actions of others in complex and strategic settings is an important facet of intelligence that guides our interactions---from walking in crowds to negotiating multi-party deals.  Recovering such behavior from merely a few observations is an important and challenging machine learning task.


While mature computational frameworks for decision-making have been developed to \textbf{prescribe} the behavior that an agent {\em should} perform, such frameworks are often ill-suited for \textbf{predicting} the behavior that an agent {\em will} perform.  Foremost, the standard assumption of decision-making frameworks that a criteria for preferring actions (\emph{e.g.}, costs, motivations and goals) is known \textit{a priori} often does not hold.
Moreover, real behavior is typically not consistently optimal or completely rational; it may be influenced by factors that are difficult to model or subject to various types of error when executed. Meanwhile, the standard tools of statistical machine learning (\emph{e.g.}, classification and regression) may be equally poorly matched to modeling purposeful behavior; an agent's goals often succinctly, but implicitly, encode a strategy that would require tremendous amounts of data to learn.

A natural approach to mitigate the complexity of recovering a full strategy for an agent is to consider identifying a compactly expressed utility function that \emph {rationalizes} observed behavior: that is, identify rewards for which the demonstrated behavior is optimal and then leverage these rewards for future prediction. Unfortunately, the problem is fundamentally ill-posed: in general, many reward functions can make behavior seem rational, and in fact, the trivial, everywhere $0$ reward function makes \textbf{all} behavior appear rational \cite{ng2000algorithms}. Further, after removing such trivial reward functions, there may be \textbf{no} reward function for which the demonstrated behavior is optimal as agents may be imperfect and the real world they operate in may be only approximately represented.

In the single-agent decision-theoretic setting, inverse optimal control methods have been used to bridge this gap between the prescriptive frameworks and predictive applications~\cite{abbeel2004,ratliff2006,ziebart2008,ziebart2010}. Successful applications include learning and prediction tasks in personalized vehicle route planning~\cite{ziebart2008}, robotic crowd navigation~\cite{henry2010}, quadruped foot placement and grasp selection~\cite{ratliff2009}.  A reward function is learned by these techniques that both explains demonstrated behavior and approximates the optimality criteria of prescriptive decision-theoretic frameworks.

As these methods only capture a single reward function and do not reason about competitive or cooperative motives, inverse optimal control proves inadequate for modeling the strategic interactions of multiple agents.  In this paper, we consider the game-theoretic concept of regret as a necessary stand-in for the optimality criteria of the single-agent work.  As with the inverse optimal control problem, the result is fundamentally ill-posed. We address this by requiring that for any utility function linear in known features, our learned model must have no more regret than that of the observed behavior.  We demonstrate that this requirement can be re-cast as a set of equivalent convex constraints that we denote the \emph{inverse correlated equilibrium } (ICE) polytope.

As we are interested in the effective prediction of behavior, we will use a maximum entropy criteria to select behavior from this polytope.  We demonstrate that optimizing this criteria leads to mini-max optimal prediction of behavior subject to approximate rationality.  We consider the dual of this problem and note that it generalizes the traditional log-linear maximum entropy family of problems~\cite{della2002inducing}.  We provide a simple and computationally efficient gradient-based optimization strategy for this family and show that only a small number of observations are required for accurate prediction and transfer of behavior.  We conclude by considering a matrix routing game and compare the ICE approach to a variety of natural alternatives.




Before we formalize imitation learning in matrix games, motivate our assumptions and describe and analyze our approach, we will review the game-theoretic notions of regret and the correlated equilibrium.

\section{Game Theory Background}
Matrix games are the canonical tool of game theorists for representing
strategic interactions
ranging from illustrative toy problems, such as the ``Prisoner's Dilemma" and
the ``Battle of the Sexes" games, to important negotiations, collaborations,
and auctions.
%
In this work, we employ a class of games with payoffs or utilities that are linear functions of features defined over the outcome space.
\begin{definition}
A \defword{linearly parameterized normal-form game}, or \defword{matrix game}, $\Gamma = (N,\A,F)$, is composed of: a finite set of \defword{players}, $N$; a set of \defword{joint-actions} or \defword{outcomes}, $\A = \times_{i\in N} A_i$, consisting of a finite set of \defword{actions} for each player, $A_i$; a set of \defword{outcome features}, $F = \{\theta^i_a\in\R^K\}$ for each outcome that induce a \defword{parameterized utility function}, $u_i(a|w) = \trans{\theta^i_a}w$ -- the reward for player $i$ achieving outcome $a$ w.r.t.\ \defword{utility weights} $w$.
 \end{definition}
For notational convenience, we let $\ami$ denote the vector $a$ excluding
component $i$ and let $\Ami=\times_{j\ne i,j\in N}A_i$ be the set of such
vectors.

In contrast to standard normal-form games where the utility functions
for game outcomes are known, in this work we assume that
``true" utility weights, $w^*$, which govern observed behavior,
are unknown.
This allows us to model real-world scenarios where a cardinal utility is not available or is subject to personal taste.




We model the players with a distribution $\sigma\in\simplex\A$ over the game's joint-actions.  Coordination between players can exist, thus, this distribution need not factor into independent strategies for each player.  Conceptually, a signaling mechanism, such as a traffic light, can be thought to sample a joint-action from $\sigma$ and communicate to each player $a_i$, its portion of the joint-action.  Each player can then consider deviating from $a_i$ using a \defword{modification function}, $\func{f_i}{A_i}{A_i}$~\cite{agt4}.

%
The \defword{switch modification function}, for instance,
\begin{align}
\switch{i}{x}{y}{a_i} = \left\{\begin{array}{cl}%
y & \mbox{if $a_i = x$} \\
a_i & \mbox{otherwise}
\end{array}\right.
\end{align}
substitutes action $y$ for recommendation $x$.

\defword{Instantaneous regret} measures how much a player would benefit from a particular modification function when the coordination device draws joint-action $a$,
\begin{align}
\regret{i}{a}{f_i}{w} & = u_i(f_i(a_i),\ami|w) - u_i(a|w) \\
& = \trans{\left[\theta^{i}_{f_i(a_i),\ami} - \theta^{i}_{a_i,\ami}\right]}w \\
& = \rvec{i}{f_i\mathrm{T}}{a}w.
\end{align}

Players 
do not have knowledge of the complete joint-action; 
thus, each must reason about the \defword{expected regret} with respect to a modification function,
\begin{align}
\trans{\sigma}\Rmat{i}{f_i}w & = \mathbb E_{a\sim\sigma}\left[\regret{i}{a}{f_i}{w}\right]\\
& = \sum_{a\in\A}\sigma_{a}\rvec{i}{f_i\mathrm{T}}{a}w.
\end{align}


It is helpful to consider regret with respect to a class of modification functions.  Two classes are particularly important for our discussion.  First, \defword{internal regret} corresponds to the set of modification functions where a single action is replaced by a new action, $\inti = \{\switch{i}{x}{y}{\cdot} : \forall x,y\in A_i\}$.  Second, \defword{swap regret} corresponds to the set of all modification functions, $\swapi = \{ f_i \}$. 
We denote $\Phi=\cup_{i\in N}\;\Phi_i$.

%
%


The \defword{expected regret with respect to $\Phi$} and outcome distribution
$\sigma$,
\begin{align}
\Regret{\Phi}{\sigma}{w} = \max_{f_i\in\Phi}\mathbb E_{a\sim\sigma}\left[\regret{i}{a}{f_i}{w}\right],
\end{align}
is important for understanding the incentive to deviate from, and hence the stability of, the specified behavior. 
The most general modification class, $\swap$, leads to the notion of \defword{$\varepsilon$-correlated equilibrium} \cite{osborne1994course}, in which $\sigma$ satisfies $\Regret{\swap}{\sigma}{w^*} \le \varepsilon$.
Thus, regret can be thought of as a substitute for utility when assessing
the optimality of behavior in multi-agent settings.

\section{Imitation Learning in Matrix Games}
We are now equipped with the tools necessary to introduce our approach for imitation learning in multi-agent settings.
As input, we observe a sequence of outcomes, $\{a_m\}_{m=1}^M$, sampled from $\truesigma$, the \defword{true behavior}.  We denote the empirical distribution of this sequence, $\demonsigma$, the \defword{demonstrated behavior}.  We aim to learn
a \defword{predictive behavior} distribution, $\predsigma$ from these demonstrations.  Moreover, we would like our learning procedure to extract the motives and intent for the behavior so that we may imitate the players in similarly structured, but unobserved games.

Imitation appears hard barring further assumptions.  In particular, if the agents are unmotivated or their intentions are not coerced by the observed game, there is little hope of recovering principled behavior in a new game.  Thus, we require some form of rationality.

\subsection{Rationality Assumptions}

We say that agents are \emph{rational} under their true preferences when they are indifferent between
$\predsigma$ 
and their true behavior if and only if $\Regret{\Phi}{\predsigma}{w^*} \le \Regret{\Phi}{\truesigma}{w^*}$.

As agents' true preferences $w^*$ are unknown to the observer, we must consider an encompassing assumption that requires any behavior that we estimate to satisfy this property for all possible utility weights, or
\begin{align}
\forall w\in\R^K, \; \Regret{\Phi}{\predsigma}{w} \le \Regret{\Phi}{\truesigma}{w}.
\end{align}
Any behavior achieving this restriction, \emph{strong rationality}, is also rational, and, by virtue of the contrapositive, we see that
unless we have additional information regarding the agents' true preferences, we must assume this strong assumption or we risk violating rationality.
\begin{lemma}
If strong rationality does not hold for alternative behavior $\predsigma$ then there exist agent utilities such that they would prefer $\truesigma$ to $\predsigma$.
\end{lemma}
By restricting our attention to behavior that satisfies strong rationality, at worst, agents acting according to unknown true preference $w^*$ will be indifferent between our predictive distribution and their true behavior.


\subsection{Inverse Correlated Equilibria}

Unfortunately, a direct translation of the strong rationality requirement into constraints on the distribution $\predsigma$ leads to a non-convex optimization problem as it involves products of varying utility vectors and the behavior to be estimated. Fortunately, however, we can provide an equivalent concise convex description of the constraints on
$\predsigma$ that ensures any feasible distribution satisfies strong rationality. We denote this set of equivalent constraints as the \emph{Inverse Correlated Equilibria} (ICE) polytope:
\begin{definition} [ICE Polytope] \label{def:ice}
\begin{align}
  \trans{\predsigma}{\Rmat{i}{f_i}} = &\sum_{f_j\in\Phi}\eta_{f_j}^{f_i}\trans{\demonsigma}{\Rmat{j}{f_j}}, \forall f_i\in\Phi \\
  \notag & \; \eta^{f_i}\in\simplex\Phi, \forall f_i\in\Phi;
\qquad \predsigma\in\simplex\A.
\end{align}
\end{definition}

\begin{theorem}
A distribution, $\predsigma$, satisfies the constraints above for some $\eta$ if and only if it satisfies strong rationality.  That is, $\forall w\in\R^K, \; \Regret{\Phi}{\predsigma}{w} \le \Regret{\Phi}{\demonsigma}{w}$ if and only if $\forall f_i\in\Phi, \exists \eta^{f_i}\in\simplex\Phi$ such that $\trans{\predsigma}{\Rmat{i}{f_i}} = \sum_{f_j\in\Phi}\eta_{f_j}^{f_i}\trans{\demonsigma}\Rmat{j}{f_j}$.
\label{thm:iff}
\end{theorem}

The proof of Theorem~\ref{thm:iff}
is provided in the Appendix~\cite{waugh11arXiv}.

We note that this polytope, perhaps unsurprisingly, is similar to the polytope of correlated equilibrium itself, but here is defined in terms of the behavior we observe instead of the (unknown) reward function. Given any observed behavior $\truesigma$, the constraints are feasible as the demonstrated behavior satisfies them; our goal is to choose from these behaviors without estimating a full joint-action distribution.
While the ICE polytope establishes a basic requirement for estimating rational behavior, there are generally infinitely many distributions consistent with its constraints.

\subsection{Principle of Maximum Entropy}

As we are interested in the problem of statistical prediction of strategic behavior, we must find a mechanism to resolve the ambiguity remaining after accounting for the rationality constraints. The \defword{principle of maximum entropy} provides a principled method for choosing such a distribution.  This choice leads to not only statistical guarantees on the resulting predictions, but to efficient optimization.


The Shannon \defword{entropy} of a distribution $\predsigma$ is defined as
$H(\predsigma) = -\sum_{x\in\mathcal X}\predsigma_x\log\predsigma_x$.
The \defword{principle of maximum entropy} advocates choosing the distribution with maximum entropy subject to known (linear) constraints~\cite{jaynes1957}:
\begin{align}
\truesigma_{\text{MaxEnt}} & = \argmax_{\predsigma\in\simplex{\mathcal{X}}} H(\predsigma), \quad\mbox{subject to:} \\
& \; g(\predsigma) = 0 \text{ and } h(\predsigma) \leq 0.\notag
\end{align}
The resulting log-linear family of distributions ({\em e.g.}, logistic regression, Markov random fields, conditional random fields) are widely used within statistical machine learning.
For our problem, the constraints are precisely that the distribution is in the ICE polytope, ensuring that whatever distribution is learned has no more regret than the demonstrated behavior.


Importantly, the maximum entropy distribution subject to our constraints enjoys the following guarantee:
\begin{lemma} \label{lem:maxent}
The maximum entropy ICE distribution minimizes over all strongly rational distributions the worst-case log-loss , $-\sum_{a \in \A} \truesigma_a \log \predsigma_a$, when $\truesigma$ is chosen adversarially and subject to strong rationality.
\end{lemma}
The proof of Lemma \ref{lem:maxent} follows immediately from the result of Gr\"unwald and Dawid~\yrcite{grunwald2003}.


In the context of multi-agent behavior, the principle of maximum entropy has been employed to obtain correlated equilibria with predictive guarantees in normal-form games when the utilities are known \textit{a priori}~\cite{ortiz2007}.  We will now leverage its power with our rationality assumption to select predictive distributions in games where the utilities are unknown.

\subsection{Prediction of Behavior}
Let us first consider prediction of the demonstrated behavior using the principle of maximum entropy and our strong rationality condition.  After, we will extend to behavior transfer and analyze the error introduced as a by-product of sampling $\demonsigma$ from $\truesigma$.

The mathematical program that maximizes the entropy of $\predsigma$ under strong rationality with respect to $\demonsigma$,
\begin{align}
\argmax_{\predsigma,\eta} & \; H(\predsigma), \quad \mbox{subject to:} \\
\notag & \;\trans{\predsigma}{\Rmat{i}{f_i}} {=} \sum_{f_j\in\Phi}\eta_{f_j}^{f_i}\trans{\demonsigma}\Rmat{j}{f_j}, \forall f_i\in\Phi \\
\notag & \; \eta^{f_i}\in\simplex\Phi, \forall f_i\in\Phi;\quad\quad \predsigma\in\simplex\A,
\end{align}
is convex with linear constraints, feasible, and bounded.  That is, it is simple and can be efficient solved in this form.  Before presenting our preferred dual optimization procedure, however, let us describe an approach for behavior transfer that further illustrates the advantages of this approach over directly estimating $\truesigma$.

\subsection{Transfer of Behavior}
A principal justification of inverse optimal control techniques that attempt to identify behavior in terms of utility functions is the ability to consider what behavior might result if the underlying decision problem were changed while the interpretation of features into utilities remain the same~\cite{ng2000algorithms,ratliff2006}.  This enables prediction of agent behavior in a no-regret or agnostic sense in problems such as a robot encountering novel terrain~\cite{Silver_2010_6638} as well as route recommendation for drivers traveling to unseen destinations~\cite{Ziebart2008b}.  

Econometricians are interested in similar situations, but for much different reasons.  Typically, they aim to validate a model of market behavior from observations of product sales.  In these models, the firms assume a fixed pricing policy given known demand.  The econometrician uses this fixed policy along with product features and sales data to estimate or bound both the consumers' utility functions as well as unknown production parameters, like markup and production cost~\cite{berry95,nevo01,yang09}.  In this line of work, the observed behavior is considered accurate to start with; it is not suitable for settings with limited observations.

Until now, we have considered the problem of identifying behavior in a single game. We note, however, that our approach enables behavior \emph{transfer} to games equipped with the same features.  We denote this unobserved game as $\bar\Gamma$.  As with prediction, to develop a technique for behavior transfer we assume a link between regret and the agents' preferences across the known space of possible preferences.  Furthermore, we assume a relation between the regrets in both games.

\begin{property}[Transfer Rationality]
For some constant $\kappa > 0$,
\begin{align}
\forall w,\:\bar{R}^{\bar\Phi}(\bar{\sigma}, w) \le \kappa\Regret{\Phi}{\truesigma}{w}.
\end{align}
\end{property}
Roughly, we assume that under preferences with low regret in the original game, the behavior in the unobserved game should also have low regret.  By enforcing this property, if the agents are performing well with respect to their true preferences, then the transferred behavior will also be of high quality.

As we are not privileged to know $\kappa$ and this property is not guaranteed to hold, we introduce a slack variable to allow for violations of the strong rationality constraints to guaranteeing feasibility. Intuitively, the \emph{transfer-ICE polytope} we now optimize over requires that for any linear reward function and for every player, the predicted behavior in a new game must have no more regret than demonstrated behavior does in the observed game using the same parametric form of reward function.  The corresponding mathematical program is:
\begin{align}
  \max_{\predsigma,\eta,\nu} & \; H(\predsigma) - C\nu,\quad\mbox{subject to:}\\
  \notag & \; \trans{\predsigma}{\baRmat{i}{f_i}} - \sum_{f_j\in\Phi}\eta_{f_j}^{f_i}\trans{\demonsigma}\Rmat{j}{f_j} \le \nu, \forall f_i\in\bar\Phi \\
  \notag & \; \sum_{f_j\in\Phi}\eta_{f_j}^{f_i}\trans{\demonsigma}\Rmat{j}{f_j} - \trans{\predsigma}{\baRmat{i}{f_i}} \le \nu, \forall f_i\in\bar\Phi \\
  \notag & \; \eta^{f_i}\in\simplex\Phi, \forall f_i\in\bar\Phi;
\; \; \predsigma\in\simplex\A;
\; \; \nu \ge 0.
\end{align}
In the above formulation, $C > 0$ is a slack penalty parameter, which allows us to choose the trade-off between obeying the rationality constraints and maximizing the entropy.  Additionally, we have omitted $\kappa$ above by considering it intrinsic to $R$.

We observe that this program is almost identical to the behavior prediction program introduced above.  We have simply made substitutions of the regret matrices and modification sets in the appropriate places.  That is, if $\bar\Gamma = \Gamma$, we recover prediction with a slack.

Given $\predsigma$ and $\nu$, we can bound the violation of the strong rationality constraint for any utility vector.
\begin{lemma}
If $\predsigma$ violates the strong rationality constraints in the slack formulation by $\nu$ then for all $w$
\begin{align}
\Regret{\Phi}{\predsigma}{w} \le \Regret{\Phi}{\demonsigma}{w} + \nu\norm{w}_1.
\end{align}
\end{lemma}

One could choose to institute multiple slack variables, say one for each $f_i\in\bar\Phi$, instead of a single slack across all modification functions.  Our choice is motivated by the interpretation of the dual multipliers presented in the next section.  There, we will also address selection of an appropriate value for $C$.

\section{Duality and Efficient Optimization}
In this section, we will derive, interpret and describe a procedure for optimizing the dual program for solving the MaxEnt ICE problem.  We will see that the dual multipliers can be interpreted as utility vectors and that optimization in the dual has computational advantages.  We begin by presenting the dual of the transfer program.
\begin{align}
\notag \min_{\alpha,\beta,\xi} & \; \sum_{f_i\in\bar\Phi} \max_{f_j\in\Phi}\left[\trans{\demonsigma}\Rmat{j}{f_j}(\alpha^{f_i}-\beta^{f_i})\right] + \log Z(\alpha,\beta)\\
\notag\mbox{subject to:} & \; \xi + \sum_{f_i\in\bar\Phi}\sum_{k=1}^K \alpha^{f_i}_k + \beta^{f_i}_k = C, 
\; \alpha,\beta, \xi \ge 0.
\end{align}
where $Z(\alpha,\beta)$ is the partition function,
\begin{align}
\notag Z(\alpha,\beta) = \sum_{a\in\bar\A}\exp\left(-\sum_{f_i\in\bar\Phi} \barvec{i}{f_i\mathrm{T}}{a}(\alpha^{f_i} - \beta^{f_i})\right).
\end{align}

Removing the equality constraint is equivalent to disallowing any slack.  We derive the dual in the appendix~\cite{waugh11arXiv}.


For $C > 0$, the dual's feasible set has non-empty interior and is bounded.  Therefore, by Slater's condition, strong duality holds -- there is no duality gap.  In particular, we can use a dual solution to recover $\predsigma$.
\begin{lemma}
Given a dual solution, $(\alpha,\beta)$, we can recover the primal solution, $\predsigma$.  Specifically,
\begin{align}
\predsigma_a = \exp\left(-\sum_{f_i\in\bar\Phi} \barvec{i}{f_i\mathrm{T}}{a}(\alpha^{f_i} - \beta^{f_i})\right) / Z(\alpha,\beta).
\end{align}
\label{lemma:primsol}
\end{lemma}
Intuitively, the probability of predicting an outcome is small if that outcome has high regret.

\begin{algorithm}[tb]
   \caption{Dual MaxEnt ICE}
   \label{alg:da}
\begin{algorithmic}
  \STATE {\bfseries Input:} $T, \gamma, C > 0$, $R, \bar{R}, \Phi$ and $\bar{\Phi}$
  \STATE $\forall f_i\in\bar\Phi,\; \alpha^{f_i}, \beta^{f_i} \leftarrow 1/(|\bar\Phi|K+1)$
  \FOR{$t$ from $1$ to $T$}
  \STATE /* compute the gradient */
  \STATE $\forall a\in\bar\A,\; z_a \leftarrow \exp\left(-\sum_{f_i\in\bar\Phi} \barvec{i}{f_i\mathrm{T}}{a}(\alpha^{f_i} - \beta^{f_i})\right)$
  \STATE $Z \leftarrow \sum_{a\in\bar\A} z_a$
  \FOR{$f_i \in \bar\Phi$}
  \STATE $f_j^* \leftarrow \argmax_{f_j\in\Phi}\trans{\demonsigma}\Rmat{j}{f_j}(\alpha^{f_i} - \beta^{f_i})$
  \STATE $g^{f_i} \leftarrow \trans{\demonsigma}\Rmat{j^*}{f^*_j} - \sum_{a\in\bar\A}z_a\barvec{i}{f_i}{a}/Z$
  \ENDFOR
  \STATE  /* descend and project */
  \STATE $\gamma_t \leftarrow \gamma/\sqrt{t}$
  \STATE $\rho \leftarrow 1 + \sum_{f_i,k}\alpha^{f_i}_k\exp(-\gamma_t g^{f_i}_k) + \beta^{f_i}_k\exp(\gamma_t g^{f_i}_k)$
  \STATE $\forall f_i\in\bar\Phi, k\in K,\; \alpha^{f_i}_k \leftarrow C\alpha^{f_i}_k\exp(-\gamma_t g^{f_i}_k)/\rho$
  \STATE $\forall f_i\in\bar\Phi, k\in K,\; \beta^{f_i}_k \leftarrow C\beta^{f_i}_k\exp(\gamma_t g^{f_i}_k)/\rho$
  \ENDFOR
  \STATE {\bfseries return} $(\alpha,\beta)$
\end{algorithmic}
\label{alg:dual}
\end{algorithm}

In general, the dual multipliers are utility vectors associated with each modification function in $\bar\Phi$.  Under the slack formulation, there is a natural interpretation of these variables as a single utility vector.
Given a dual solution, $(\alpha,\beta)$ with slack penalty $C$, we choose
\begin{align}
\lambda^{f_i} & = \alpha^{f_i} - \beta^{f_i}, \\
\pi^{f_i} & = \frac{1}{C}\sum_{k=1}^K \alpha^{f_i}_k + \beta^{f_i}_k, \mbox{and} \\
\predw & = \sum_{f_i\in\bar\Phi}\pi^{f_i}\lambda^{f_i}.
\end{align}
That is, we can associate with each modification function a probability, $\pi^{f_i}$, and a utility vector, $\lambda^{f_i}$.  Thus, a natural estimate for $\predw$ is the expected utility vector.  Note, $\sum_{f_i\in\bar\Phi}\pi^{f_i}$ need not sum to one.  The remaining mass, $\xi$, is assigned to the zero utility vector.

The above observation implies that introducing a slack variable coincides with bounding the $L_1$ norm of the utility vectors under consideration by $C$.  This insight suggests that we choose $C \ge \norm{w^*}_1$, if possible, as smaller values of $C$ will exclude $w^*$ from the feasible set.  If a bound on the $L_1$ norm is not available, we may solve the prediction problem on the observed game without slack and use $\norm{\predw}_1$ as a proxy.

The dual formulation of our program has important inherent computational advantages.  First, it is a optimization over a simple set that is particularly well-suited for gradient-based optimization, a trait not shared by the primal program.  Second, the number of dual variables, $2|\Phi|K$, is typically much fewer than the number of primal variables, $|\A|+2|\Phi|^2$.  Though the work per iteration is still a function of $|\A|$ (to compute the partition function), these two advantages together let us scale to larger problems than if we consider optimizing the primal objective. Computing the expectations necessary to descend the dual gradient can leverage recent advances in the structured, compact game representations: in particular, any graphical game with low-treewidth or finite horizon Markov game~\cite{kakade2003correlated} enables these computations to be performed in time that scales only polynomially in the number of decision makers or time-steps.  

Algorithm~\ref{alg:dual} employs exponentiated gradient descent~\cite{kivinen95} to find an optimal dual solution.  The step size parameter, $\gamma$, is commonly taken to be $\sqrt{2\log |\bar\Phi|K}/\Delta$, with $\Delta$ being the largest value in any $R^{f_i}_i$.  With this step size, if the optimization is run for $T\ge 2\Delta^2\log\left(|\bar\Phi|K\right) / \epsilon^2$ iterations then the dual solution will be within $\epsilon$ of optimal.  Alternatively, one can exactly measure the duality gap on each iteration and halt when the desired accuracy is achieved.  This is often preferred as the lower bound on the number of iterations is conservative in practice.





\section{Sample Complexity}

In practice, we do not have full access to the agents' true behavior -- if we did, prediction would be straightforward and not require our estimation technique. Instead, we can only approximate it through finite observation of play.  In real applications there are costs associated with gathering these observations and, thus, there are inherent limitations on the quality of this approximation.  In this section, we will analyze the sensitivity of our approach to these types of errors.

First, although $|\A|$ is exponential in the number of players, our technique only accesses $\demonsigma$ through products of the form $\demonsigma\Rmat{j}{f_j}$.  That is, we need only approximate these products accurately, not the distribution $\demonsigma$.  As a result, we can bound the approximation error in terms of $|\Phi|$ and $K$.
\begin{theorem}
With probability at least $1-\delta$, for any $w$, by observing $M \ge \frac{2}{\epsilon^2}\log\frac{2|\Phi|K}{\delta}$
outcomes we have $\Regret{\Phi}{\demonsigma}{w} \le \Regret{\Phi}{\truesigma}{w} + \epsilon\Delta\norm{w}_1$.
\label{thm:sample}
\end{theorem}
The proof is an application of Hoeffding's inequality and is provided in the Appendix~\cite{waugh11arXiv}.  As an immediate corollary, considering only the true, but unknown, reward function $w^*$:
\begin{corollary}
With probability at least $1-\delta$, by sampling according to the above rule, $\Regret{\Phi}{\predsigma}{w^*} \le \Regret{\Phi}{\truesigma}{w^*} + (\epsilon\Delta+\nu)\norm{w^*}_1$ for $\predsigma$ with slack $\nu$.
\label{corollary:regretsclose}
\end{corollary}
That is, so long as we assume bounded utility, with high probability we need only logarithmic many samples in terms of $|\Phi|$ and $K$ to closely approximate $\truesigma\Rmat{j}{f_j}$ and avoid a large violation of our rationality condition.

We note that choosing $\Phi = \modint$ is particularly appealing, as $|\modint| \le |N|A^2$, compared to $|\swap| \le |N|A!$.  As internal regret closely approximates swap regret, we do not lose much of the strategic complexity by choosing the more limited set, but we require both fewer observations and fewer computational resources.

\section{Experimental Results}

To evaluate our approach experimentally, we designed a simple routing game shown in Figure~\ref{fig:routing}.  
Seven drivers in this game choose how to travel home during rush hour after a long day at the office.  The different road segments have varying capacities, visualized by the line thickness in the figure, that make some of them more or less susceptible to congestion or to traffic accidents.  Upon arrival home, each driver records the total time and distance they traveled, the gas that they used, and the amount of time they spent stopped at intersections or in traffic jams -- their utility features.

\begin{figure}[h]
\begin{center}
\centerline{\includegraphics[width=\columnwidth]{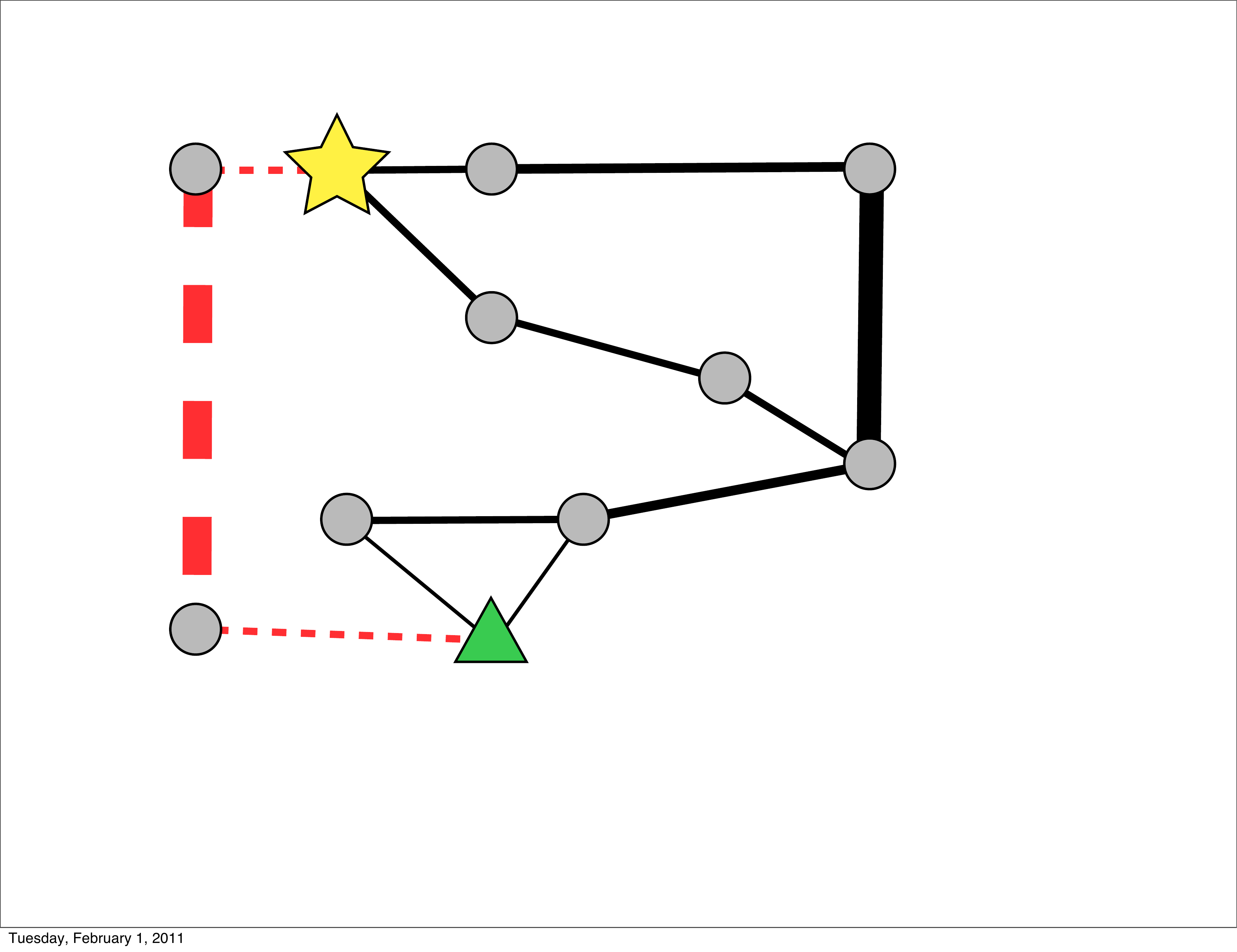}}
\caption{A visualization of the routing game.}
\label{fig:routing}
\end{center}
\vskip -0.2in
\end{figure}

In this game, each of the drivers chooses from four possible routes (solid lines in Figure \ref{fig:routing}), yielding over 16,000 possible outcomes.  We obtained an $\varepsilon$-social welfare maximizing correlated equilibrium for those drivers where the drivers preferred mainly to minimize their travel time, but were also slightly concerned with gas usage.  The demonstrated behavior $\demonsigma$ was sampled from this true behavior distribution $\truesigma$.

First, we evaluate the differences between the true behavior distribution $\truesigma$ and the predicted behavior distribution $\predsigma$ trained from observed behavior sampled from $\demonsigma$.  In Figure~\ref{fig:logloss} we compare the prediction accuracy when varying the number of observations using log-loss, $-\sum_{a\in\A} \truesigma_a \log \predsigma_a$.  The baseline algorithms we compare against are: a maximum likelihood estimate of the distribution over the joint-actions with a uniform prior, an exponential family distribution parameterized by the outcome's utilities trained with logistic regression, and a maximum entropy inverse optimal control approach \cite{ziebart2008} trained individually for each player.

\begin{figure}[h]
\begin{center}
\centerline{\includegraphics[width=\columnwidth]{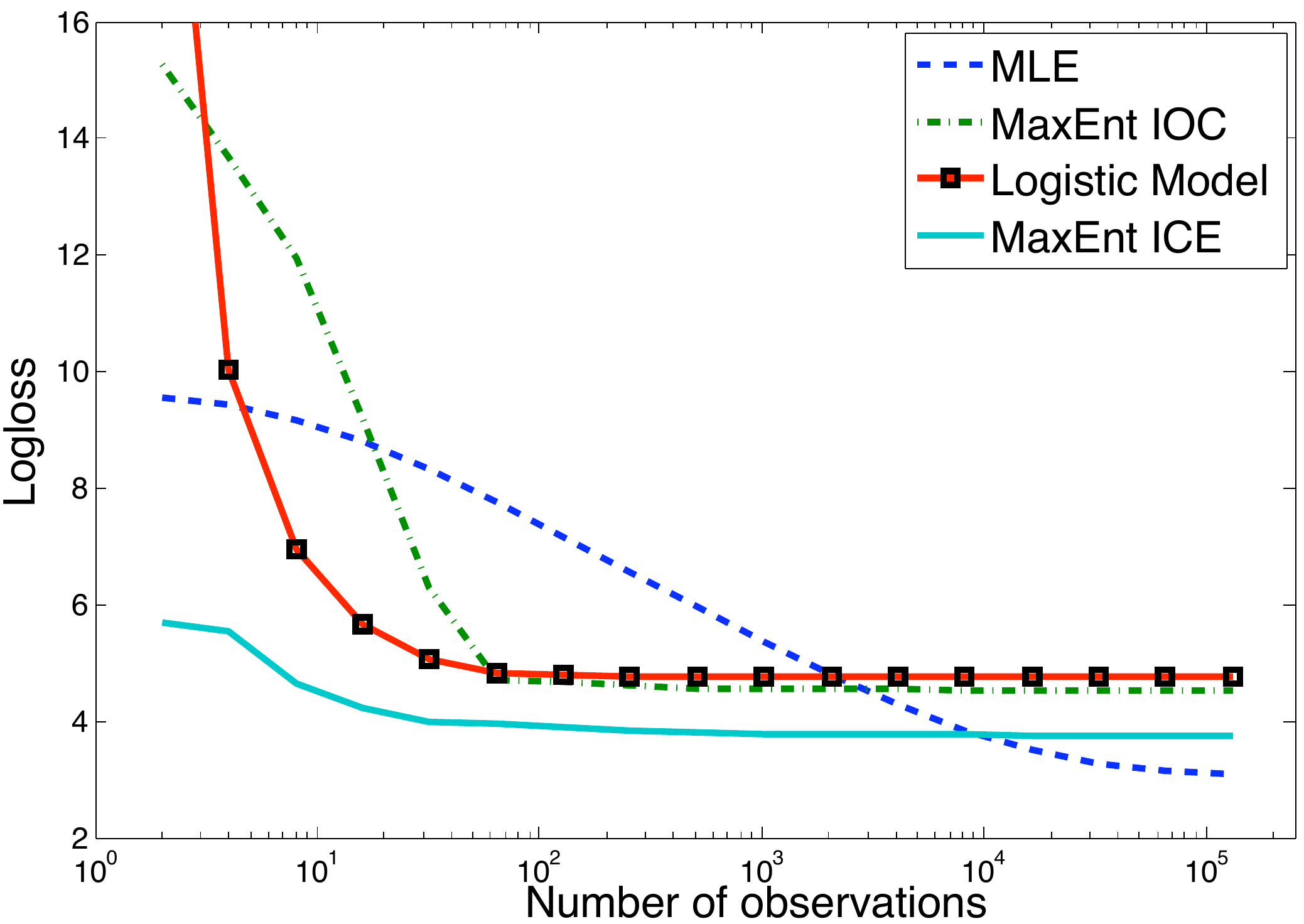}}
\caption{Prediction error (log-loss) as a function of number of observations.}
\label{fig:logloss}
\end{center}
\vskip -0.2in
\end{figure}

In Figure~\ref{fig:logloss}, we see that MaxEnt ICE predicts behavior with higher accuracy than all other algorithms when the number of observations is limited.  In particular, it achieves close to its best performance with as few at $16$ observations.  The maximum likelihood estimator eventually overtakes it, as expected since it will ultimately converge to $\truesigma$, but only after 10,000 observations, or about as many observations as there are outcomes in the game.  This experiment demonstrates that learning underlying utility functions to estimate observed behavior can be much more data-efficient for small sample sizes, and additionally, that the regret-based assumptions of MaxEnt ICE are both reasonable and beneficial in our strategic routing game setting.

Next, we evaluate behavior transfer from this routing game to four similar games, the results of which are displayed in Table~\ref{tbl:transfer}.  The first game, {\em Add Highway}, adds the dashed route to the game.  That is, we model the city building a new highway.  The second game, {\em Add Driver}, adds another driver to the game.  The third game, {\em Gas Shortage}, keeps the structure of the game the same, but changes the reward function to make gas mileage more important to the drivers.  The final game, {\em Congestion}, adds construction to the major roadway, delaying the drivers.

\begin{table}[t]
\caption{Transfer error (log-loss) on unobserved games.}
\label{tbl:transfer}
\vskip 0.15in
\begin{center}
\begin{small}
\begin{sc}
\begin{tabular}{lcc}
\hline
\abovespace\belowspace
Problem & Logistic Model & MaxEnt Ice \\
\hline
\abovespace
Add Highway & 4.177  & 3.093 \\
Add Driver & 4.060 & 3.477 \\
Gas Shortage & 3.498 & 3.137 \\
Congestion & 3.345 & 2.965 \\
\hline
\end{tabular}
\end{sc}
\end{small}
\end{center}
\vskip -0.2in
\end{table}

These transfer experiments even more directly demonstrate the benefits of learning utility weights rather than directly learning the joint-action distribution; direct strategy-learning approaches are incapable of being applied to general transfer setting.  Thus, we only compare against the Logistic Model.  We see from Table~\ref{tbl:transfer} that MaxEnt ICE outperforms the Logistic Model in all of our tests.  For reference, in these new games, the uniform strategy has a loss of approximately $6.8$ in all games, and the true behavior has a loss of approximately $2.7$.

\section{Conclusion}
In this paper, we extended inverse optimal control to multi-agent settings by combining the principle of maximum entropy with the game-theoretic notion of regret.  We observed that our formulation has a particularly appealing dual program, which led to a simple gradient-based optimization procedure.  Perhaps the most appealing quality of our technique is its theoretical and practical sample complexity.  In our experiments, MaxEnt ICE performed exceptionally well after only $0.1\%$ of the game had been observed.
\section*{Acknowledgments}
This work is supported by the ONR MURI grant N00014-09-1-1052 and by the National Sciences and Engineering Research Council of Canada (NSERC).

{
\bibliography{paper}
\bibliographystyle{icml2011}
}

\newpage

\onecolumn
\appendix
\section*{Appendix}
\subsection*{Rationality Properties and Primal Programs}

The proof of Theorem~\ref{thm:iff} relies upon the following technical 
lemmas.

\begin{lemma}
 \[
 \trans{b}{w} \le \max_{a_i\in A}\trans{a_i}{w} \Leftrightarrow \exists \lambda\in\simplex A \st \trans{b}{w} \le \lambdaT Aw.
 \]
 \label{lemma:tech1}
 \end{lemma}

\begin{proof}[Proof of Lemma~\ref{lemma:tech1}]
Given $\trans{b}{w} \le \max_{a_i\in A}\trans{a_i}{w}$, choose
\begin{align}
\lambda_i = \left\{\begin{array}{cl}%
1 & \mbox{if $a_i = \argmax_{a_i\in A} \trans{a_i}{w}$} \\
0 & \mbox{otherwise}
\end{array}\right.
\end{align}
Thus, $\trans{b}{w} \le \max_{a_i\in A}\trans{a_i}{w} = \lambdaT Aw$. \\
Given $\exists \lambda\in\simplex A \st \trans{b}{w} \le \lambdaT Aw$,
\begin{align}
& \; \trans{b}{w} \le \lambdaT Aw \\
\le & \; \sum_{a_j\in A}\lambda_{a_j}\max_{a_i\in A}\trans{a_i}{w} \\
= & \; \left[\max_{a_i\in A}\trans{a_i}{w}\right]\sum_{a_j\in A}\lambda_{a_j} \\
= & \max_{a_i\in A}\trans{a_i}{w}
\end{align}
\end{proof}

 \begin{lemma}
 \[
 \forall w\in\R^K, \trans{b}{w} \le \max_{i\in N} \trans{a_i}{w} \Leftrightarrow \exists \lambda\in\simplex A \st b = \lambdaT A.
 \]
 \label{lemma:tech2}
 \end{lemma}

\begin{proof}[Proof of Lemma~\ref{lemma:tech2}]
\begin{align}
& \; \forall w\in\R^K, \trans{b}{w} \le \max_{a_i\in A} \trans{a_i}{w} \\
\Leftrightarrow & \; \forall w\in\R^K, \exists \lambda\in\simplex A \st \trans{b}{w} \le \lambdaT Aw \\
\Leftrightarrow & \; \forall w\in\R^K, \exists \lambda\in\simplex A \st \trans{\left[b - \lambdaT A\right]}w \le 0
\end{align}
$\Leftrightarrow$ the following linear program has optimal value $0$
\begin{align}
\max_{w,t} & \; \trans{b}{w} - t \\
\notag\mbox{subject to:} & \; t \ge \trans{a_i}{w}, \forall a_i\in A.
\end{align}
The following linear feasibility problem is the dual of the above program
\begin{align}
\min_{\lambda} & \; 0 \\
\notag\mbox{subject to:} & \; b = \lambdaT A \\
\notag & \; \lambda\in\simplex A.
\end{align}
By strong duality for linear programming, the primal has value $0$ iff the dual is feasible, which is exactly when $\exists \lambda\in\simplex A \st b = \lambdaT A$.
\end{proof}

\begin{proof}[Proof of Theorem~\ref{thm:iff}]
\begin{align}
& \; \forall w\in\R^K, \; \Regret{\Phi}{\predsigma}{w} \le \Regret{\Phi}{\demonsigma}{w} \\
\Leftrightarrow & \; \forall w\in\R^K, \max_{f_i\in\Phi} \trans{\predsigma}{\Rmat{i}{f_i}w} \le \max_{f_i\in\Phi} \trans{\demonsigma}{\Rmat{i}{f_i}w} \\
\Leftrightarrow & \; \forall f_i\in\Phi, \forall w\in\R^K, \trans{\predsigma}{\Rmat{i}{f_i}w} \le \max_{f_j\in\Phi} \trans{\demonsigma}{\Rmat{j}{f_j}w} \\
\Leftrightarrow & \; \forall f_i\in\Phi, \exists \eta^{f_i}\in\simplex\Phi \st \trans{\predsigma}{\Rmat{i}{f_i}} = \sum_{f_j\in\Phi}\eta^{f_i}_{f_j} \trans{\demonsigma}{\Rmat{j}{f_j}}
\end{align}
The last step makes use of our second technical lemma.
\end{proof}

\subsection*{Derivation of the Dual Program}
The Lagrange dual is
\begin{align}
\min_{\alpha,\beta,\gamma,\delta,u,v,\xi}\max_{\predsigma,\eta,\nu} & \; -\sum_{a\in\bar\A}\predsigma_a\log\predsigma_a - C\nu - \sum_{f_i\in\bar\Phi}\left(\predsigma\baRmat{i}{f_i} - \sum_{f_j\in\Phi}\eta^{f_i}_{f_j}\trans{\demonsigma}\Rmat{j}{f_j} - \nu\right)\alpha^{f_i} \\
& \; - \sum_{f_i\in\bar\Phi}\left(\sum_{f_j\in\Phi}\eta^{f_i}_{f_j}\trans{\demonsigma}\Rmat{j}{f_j} - \predsigma\baRmat{i}{f_i} - \nu\right)\beta^{f_i}\\
& \; + \sum_{f_i\in\bar\Phi}\left(1 - \sum_{f_j\in\Phi}\eta^{f_i}_{f_j}\right)\gamma^{f_i} + \left(1 - \sum_{a\in\bar\A}\predsigma_a\right)\delta \\
& \; + \sum_{f_i\in\bar\Phi}\sum_{f_j\in\Phi}\eta^{f_i}_{f_j}u^{f_i}_{f_j} + \sum_{a\in\A}\predsigma_av_a  + \nu\xi\\
\mbox{subject to:} & \; \alpha,\beta,u,v,\xi \ge 0
\end{align}
To solve the unconstrained inner optimization, we take derivatives w.r.t.\ $\sigma$, $\eta$ and $\nu$ and set equal to $0$:
\begin{align}
& \log\predsigma_a = -1 - \sum_{f_i\in\bar\Phi}\barvec{i}{f_i}{a}(\alpha^{f_i} - \beta^{f_i}) - \delta + v_a = 0,\\
& \trans{\demonsigma}\Rmat{j}{f_j}(\alpha^{f_i}-\beta^{f_i}) - \gamma^{f_i} + u^{f_i}_{f_j} = 0, \quad\forall f_i,\in\bar\Phi, f_j\in\Phi, \mbox{~and} \\
& \xi - C + \sum_{f_i\in\bar\Phi}\alpha^{f_i} + \beta^{f_i} = 0.
\end{align}
Substituting into the Lagrangian, we get
\begin{align}
\min_{\alpha,\beta,\gamma,\delta,u,v,\xi} & \; \sum_{f_i\in\bar\Phi}\gamma^{f_i} + \delta + \exp(-1-\delta)\sum_{a\in\bar\A}\exp\left(-\sum_{f_i\in\bar\Phi}\barvec{i}{f_i}{a}(\alpha^{f_i}-\beta^{f_i}) + v_a\right) \\
\mbox{subject to:} & \; \trans{\demonsigma}\Rmat{j}{f_j}(\alpha^{f_i} - \beta^{f_i}) - \gamma^{f_i} + u^{f_i}_{f_j} = 0, \quad\forall f_i\in\bar\Phi,f_j\in\Phi\\
& \; \xi + \sum_{f_i\in\bar\Phi}\alpha^{f_i} + \beta^{f_i} = C,\\
& \; \alpha, \beta, u, v, \xi\ge 0.
\end{align}
We note that $u$ are slack variables, and that, by inspection, $v = 0$ at optimality.  Thus, an equivalent program is
\begin{align}
\min_{\alpha,\beta,\gamma,\delta,\xi} & \; \sum_{f_i\in\bar\Phi}\gamma^{f_i} + \delta + \exp(-1-\delta)\sum_{a\in\bar\A}\exp\left(-\sum_{f_i\in\bar\Phi}\barvec{i}{f_i}{a}(\alpha^{f_i}-\beta^{f_i})\right) \\
\mbox{subject to:} & \; \trans{\demonsigma}\Rmat{j}{f_j}\lambda^{f_i} \le \gamma^{f_i}, \quad\forall f_i\in\bar\Phi,f_j\in\Phi\\
& \; \xi + \sum_{f_i\in\bar\Phi}\alpha^{f_i} + \beta^{f_i} \le C,\\
& \; \alpha, \beta, \xi\ge 0.
\end{align}
We eliminate $\delta$ by setting its partial derivative to $0$, solving for $\delta$
\begin{align}
\delta = \log\left(\sum_{a\in\bar\A}\exp\left(-\sum_{f_i\in\bar\Phi}\barvec{i}{f_i}{a}(\alpha^{f_i}-\beta^{f_i})\right)\right) - 1
\end{align}
and substituting back into the objective
\begin{align}
\min_{\alpha,\beta,\gamma,\xi} & \; \sum_{f_i\in\bar\Phi}\gamma^{f_i} + \log\left(\sum_{a\in\bar\A}\exp\left(-\sum_{f_i\in\bar\Phi}\barvec{i}{f_i}{a}(\alpha^{f_i}-\beta^{f_i})\right)\right) - 1 \\
\mbox{subject to:} & \; \trans{\demonsigma}\Rmat{j}{f_j}(\alpha^{f_i}-\beta^{f_i}) \le \gamma^{f_i}, \quad\forall f_i\in\bar\Phi,f_j\in\Phi\\
& \; \xi + \sum_{f_i\in\bar\Phi}\alpha^{f_i} + \beta^{f_i} \le C,\\
& \; \alpha, \beta, \xi\ge 0.
\end{align}
By inspection, at optimality, $\gamma^{f_i} = \max_{f_j\in\Phi}\trans{\demonsigma}\Rmat{j}{f_j}(\alpha^{f_i}-\beta^{f_i})$.  Thus an equivalent program is
\begin{align}
\min_{\alpha,\beta,\xi} & \; \sum_{f_i\in\bar\Phi}\left[\max_{f_j\in\Phi}\trans{\demonsigma}\Rmat{j}{f_j}(\alpha^{f_i}-\beta^{f_i})\right] + \log\left(\sum_{a\in\bar\A}\exp\left(-\sum_{f_i\in\bar\Phi}\barvec{i}{f_i}{a}(\alpha^{f_i}-\beta^{f_i})\right)\right) - 1 \\
& \; \xi + \sum_{f_i\in\bar\Phi}\alpha^{f_i} + \beta^{f_i} = C,\\
& \; \alpha, \beta, \xi\ge 0.
\end{align}

\begin{proof}[Proof of Lemma~\ref{lemma:primsol}]
In the derivation of the dual program, we observed that at optimality
\begin{equation}
\log\predsigma_a = -1 - \sum_{f_i\in\bar\Phi}\barvec{i}{f_i}{a}(\alpha^{f_i}-\beta^{f_i}) - \delta + v_a = 0.
\end{equation}
Noting $v = 0$ and substituting for the optimal $\delta$, we get
\begin{equation}
\log\predsigma_a = -\sum_{f_i\in\bar\Phi}\barvec{i}{f_i}{a}(\alpha^{f_i} - \beta^{f_i}) - \log\left(\sum_{a'\in\bar\A}\exp\left(-\sum_{f_j\in\bar\Phi}\barvec{j}{f_j}{a'}(\alpha^{f_j}-\beta^{f_j})\right)\right).
\end{equation}
All that remains is to exponentiate both sides.
\end{proof}

\subsection*{Sample Complexity}
\begin{proof}[Proof of Theorem~\ref{thm:sample}]
\begin{align}
P\left( \max_{f_i\in\Phi,k\in K} |\demonsigma\Rmat{i}{f_i} - \truesigma\Rmat{i}{f_i}|_k \ge \epsilon\Delta M \right) & \le
P\left( \bigcup_{f_i\in\Phi,k\in K} |\demonsigma\Rmat{i}{f_i} - \truesigma\Rmat{i}{f_i}|_k \ge \epsilon\Delta M \right) \\
& \le \sum_{f_i\in\Phi,k\in  K} P\left( |\demonsigma\Rmat{i}{f_i} - \truesigma\Rmat{i}{f_i}|_k \ge \epsilon\Delta M \right) \\
& \le \sum_{f_i\in\Phi,k\in  K} 2\exp\left(\frac{-\epsilon^2M}{2} \right) \\
& = 2|\Phi|K\exp\left(\frac{-\epsilon^2M}{2} \right) \\
& \le \delta
\end{align}
We use the union bound in step 2, and Hoeffding's inequality in step 3.  Solving for $M$, we get our result
\begin{align}
M \ge \frac{2}{\epsilon^2}\log\frac{2|\Phi|K}{\delta}.
\end{align}
\end{proof}

\begin{proof}[Proof of Corollary~\ref{corollary:regretsclose}]
We have $\forall w, \Regret{\Phi}{\predsigma}{w} \le \Regret{\Phi}{\demonsigma}{w} + \nu\norm{w}_1$, where $\nu$ depends on the choice of the slack's penalty.  Thus, we have $\Regret{\Phi}{\predsigma}{w^*} \le \Regret{\Phi}{\demonsigma}{w^*} + \nu\norm{w}_1 \le \Regret{\Phi}{\truesigma}{w^*} + (\epsilon\Delta + \nu)\norm{w^*}_1$ with probability at least $1-\delta$, so long as $M$ is as large as Theorem~\ref{thm:sample} deems.  We can make $\nu$ as small as we like by increasing the slack penalty.
\end{proof}

\end{document}